\documentclass[twoside,11pt]{article}
\usepackage{jmlr2e_draft}
   
\usepackage{amsmath, amssymb}
\usepackage{stmaryrd}
\usepackage{epsfig,calc, graphicx}
\usepackage{wasysym}
\usepackage[usenames]{color}
\usepackage{ifthen}
\newboolean{onecol}
\setboolean{onecol}{true}

\newcommand\ip[2]{\langle #1, #2\rangle}
\newcommand\absip[2]{|\langle #1, #2\rangle|}

\newcommand\ceil[1]{\left\lceil #1 \right\rceil}

\newcommand\smallS{s}

\newcommand\dico{\Phi}
\newcommand\atom{\varphi}

\newcommand\Sset{{\mathbb{S}}}
\newcommand\Bc{{B^c}}
\newcommand\noise{\eta}
\newcommand\nsig{\rho}

\newcommand\argmax{\operatorname{argmax}}
\newcommand\argmin{\operatorname{argmin}}
\newcommand\SNR{\operatorname{SNR}}
\newcommand{\I}{{\mathbb{I}}}

\newcommand{\R}{{\mathbb{R}}}

\renewcommand{\P}{{\mathbb{P}}}
\newcommand{\E}{{\mathbb{E}}}

\title{Average performance of Orthogonal Matching Pursuit (OMP) for sparse approximation}

\author{\name Karin Schnass \email karin.schnass@uibk.ac.at\\
\addr  Department of Mathematics\\
University of Innsbruck\\
 Technikerstra\ss e 13\\
  6020 Innsbruck, Austria}

\editor{}

\begin{document}

\maketitle 

\begin{abstract}
We present a theoretical analysis of the average performance of OMP for sparse approximation. 
For signals that are generated from a dictionary with $K$ atoms and coherence $\mu$ and coefficients corresponding
to a geometric sequence with parameter~$\alpha<1$, we show that OMP is successful with high probability as long as the
sparsity level $S$ scales as $S\mu^2 \log K \lesssim 1-\alpha $. This improves by an order of magnitude over worst case results and shows that OMP and its famous competitor Basis Pursuit outperform each other depending on the setting.
\end{abstract}

\begin{keywords} sparse approximation; Orthogonal Matching Pursuit; average case analysis; decaying coefficients;
\end{keywords}

\section{Introduction}
In sparse approximation the goal is to approximate a given signal $y\in \R^d$ by a linear combination 
of a small number $S\ll d$ of elements $\atom_i\in \R^d$, called atoms, out of a given larger set, such as basis or a frame,
called the dictionary. Storing the normalised atoms as columns in the dictionary matrix $\dico =(\atom_1 \ldots, \atom_K)$,
and denoting the restriction to the columns indexed by a set $I$ by $\dico_I$,
we can write informally,
\begin{align}
\mbox{find} \quad y \approx {\sum}_{k \in I} \atom_k x_k  = \dico_I x_I \quad \mbox{s.t.} \quad |I| = S\ll d
\end{align}
Finding the smallest error for a given sparsity level $S$ and the corresponding support set $I$, which determines $x_I$ via $x_I = \dico_I^\dagger y$, where $ \dico_I^\dagger$ is the Moore-Penrose pseudo inverse, becomes an NP-hard problem in general unless the dictionary is an orthonormal system. In this case thresholding, meaning choosing as $I$ the indices
of the atoms having the $S$-largest inner products with the signal in magnitude, will succeed. For all other cases, one had to find algorithms which are more efficient, if less optimal than an exhaustive search through all possible supports sets $I$ with subsequent projection $P(\dico_I)y:=\dico_I \dico_I^\dagger y$. The two most investigated directions are greedy methods and convex relaxation techniques - the two golden classics being Orthogonal Matching Pursuit (OMP), \cite{parekr93}, and Basis Pursuit (BP), \cite{donoho:bp}, respectively.\\
OMP finds the support iteratively, adding the index of the atom which has the largest absolute inner product with the residual and updating the residual. So initialising $r_0 = y$, $J_0= \emptyset$, it
\begin{align*}
&\mbox{finds} \quad j = \argmax_k \absip{r_i}{\atom_k}\quad \mbox{and}\\
&\mbox{updates}\quad J_{i+1} = J_i \cup \{j\} \quad \mbox{resp.} \quad r_{i+1} = y - P(\dico_{J_{i+1}})y,
\end{align*}
until a stopping criterion is met, such as reaching the desired number of iterations or the size of the residual/largest inner product being sufficiently small.\\
The Basis Pursuit principle, on the other hand, prescribes finding the minimiser of the convex programme
\begin{align}
\hat x = \argmin_{x: y = \dico x} \|x\|_1,
\end{align} 
and choosing $I$ as the index set of the $S$-largest entries of $\hat x$ in magnitude.\\
The interesting question concerning both schemes is when they are successful. So
assuming that the signal $y$ is known to be S-sparse, meaning $y = \dico_I x_I$ with $|I|=S$, when can they recover the support $I$.
It was first studied in \cite{Tropp:greed, bp:fuchs} and for dictionaries with coherence $\mu:=\max_{j\neq k}\absip{\atom_j}{\atom_k}$
a sufficient condition for both schemes to succeed is that $2S\mu < 1$, which is relaxed in comparison to the sufficient condition for thresholding $2S\mu < \frac{\min_{k\in I} |x_k|}{\max_{k\in I} |x_k|}$, but still quite restrictive, especially considering the much better performance in practice. This led to the investigation of the average performance when modelling the signals as 
\begin{align}\label{noiseless_model}
y =  {\sum}_k \sigma_k c_k \atom_{p(k)},
\end{align}
where $(\sigma_k)_k$ is a Rademacher sequence, the coefficient sequence $c$ is non-increasing, $c_k\geq c_{k+1} \geq 0$, and $c_k = 0$ for $k>S$ and $p$ is some permutation such that the support $I = \{p(1),\ldots ,p(S)\})$ satisfies $\delta_I := \| \dico_I^\star \dico_I - \I_d \|_{2,2} \leq \frac{1}{2}$, where for a matrix $A$ the transpose is denoted by $A^\star$. \\
It was shown that BP recovers the true support except with probability $2K^{1-2m}$ as long as $16 \mu^2 S \cdot m\log K \leq 1$,\cite{tr08}\footnote{The theorem actually considers Steinhaus instead of Rademacher sequences. However, the proof for Rademacher sequences is exactly the same; simply use Hoeffding's inequality instead of the complex Bernstein inequality. Also beware the buggy support condition in model M1.}, and that thresholding succeeds except with probability $2K^{1-2m}$ as long as $32 \mu^2 S \cdot m\log K \leq \frac{\min_{k\in I} |x_k|}{\max_{k\in I} |x_k|}$, \cite{scva07}\footnote{The improved constant presented here is due to the fact that also for Rademacher sequences $c_0=\frac{1}{2}$.}. 
The fact that for OMP a similar result could only be found in a multi-signal scenario, \cite{grrascva08}, started to give OMP the reputation of being weaker than BP. \\
This was further increased by the advent of Compressed Sensing (CS), \cite{do06cs}, which can be seen as sparse approximation with design freedom for the dictionary. While for BP-type schemes in combination with randomly chosen dictionaries strong results appeared very early, \cite{cata06, carota06}, comparable results for OMP and its variants took longer to develop and are weaker in general, \cite{gitr07,romp}. Still, thanks to its computational advantages and flexibility, e.g. concerning the stopping criteria, OMP remained popular in signal processing - the only difference being that users had a defensive statement a la 'of course BP will perform even better' ready at all times. \\
{\bf Contribution:} Here we will provide the long missing analysis of the average performance of OMP and show that on average neither BP nor OMP are stronger, but confirm folklore wisdom, that OMP works better for signals with decaying coefficients while BP is better for equally sized coefficients. The idea that the performance of OMP improves for decaying coefficients has already been used in \cite{hedrso16} and the simplified result states that if the sorted absolute coefficients form a geometric sequence with decay $\alpha\leq \frac{1}{2}$, then OMP is guaranteed to succeed for all sparsity levels $S$ with $S\mu <1$. Replacing certainty with high probability we will relax this bound by an order of magnitude to $S \mu^2 \sqrt{\log K} \lesssim 1 - \alpha $, for $\alpha < 1$. \\ 
{\bf Organisation:} We will first address full support recovery in the noiseless case and then extend this result to partial support recovery in the noisy case. In Section~\ref{sec_sim} we will conduct two experiments showing that our theoretical results accurately predict the average performance, before finally discussing our results and future work.

%
%Convex relaxation methods like BP super popular with compressed sensing,
%which is basically sparse recovery but the dictionary can be designed.
%very soon proofs that recovery works for any S-sparse signal via BP, if 
%the dictionary satisfies the restricted isometry property, which random dictionaries
%satisfy with high probability. Comparable results for OMP and its variants took longer to appear
%and are weaker in general, so that quite soon using OMP 

\section{Noiseless Case} \label{sec_noiseless}
We start with the simple case of signals following the model in \eqref{noiseless_model}. 
Note that from \cite{tr08} we know that for a randomly chosen subset $I$ (permutation $p$) the condition
$\delta_I \leq \frac{1}{2}$ is satisfied with high probability as long as $\mu^2 S \log K \lesssim 1$.

%%%% theorem
\begin{theorem} Assume that the signals follow the model in \eqref{noiseless_model} and that
for $i \leq S$ the coefficients satisfy $c_{i+t}/c_i \leq 1-\frac{\lambda}{S}$ for $t,\lambda > 0$. 
Then, except with probability $2SK^{1-2m}$, OMP will recover the full support as long as 
\begin{align}\label{noiseless_cond_final}
\left( t\ceil{\frac{S}{\lambda}} + \sqrt{\frac{mtS\log K}{\lambda}} + 1 \right) S \mu^2 \leq \frac{1}{13}.
\end{align}
\end{theorem}
%%% end theorem
%
%%%% proof idea 
Before presenting the proof, we want to provide some background information on
the ideas used for proving success of OMP and the difficulties associated with an average
case analysis.
A necessary and sufficient condition for a step of OMP to succeed is that for the current (correct) sub-support $J\subset I$
we have 
\begin{align}
\max_{i \in I} \absip{\atom_i}{r_J} > \max_{k \notin I} \absip{\atom_k}{r_J}.
\end{align}
Thus a sufficient condition for OMP to fully recover the support is, that for all possible sub-supports $J$ 
the missing atom which has the largest coefficient $\atom_{i_J}$ satisfies
\begin{align}
\absip{\atom_{i_J}}{r_J} > \max_{k \notin I} \absip{\atom_k}{r_J}.
\end{align}
If the coefficients have random signs then for all $k, J$ the inner products
should concentrate around their expectation, 
\begin{align}
\absip{\atom_i}{r_J}^2 \rightsquigarrow \sum_{k\in I/J} \absip{\atom_i}{[\I_d - P(\dico_J)] \atom_k}^2 c^2_k  \rightsquigarrow c^2_i \pm  \mu^2 \|c_{I/J}\|^2_2,
\end{align}
so a condition of the form $S\mu^2 \lesssim 1$ should ensure success with high probability.
The problem is that there are $2^S$ sub-supports $J$ for which we need to have this concentration.
So taking a union bound for the probability of not enough concentration over all sub-supports, we get back the worst case condition but with a non-zero failure probability.
The immediate conclusion is that in order to get a useful average case result, we have to reduce the number of intermediate
supports that we need to control. For equally sized coefficients this is impossible, since the random signs determine
the order of the absolute inner products. 
However, if the coefficients exhibit some decay, there is a natural order and it is more likely that atoms with large coefficients are picked first. For instance, with sufficient
decay it might happen that the atom with the second largest coefficient is picked before that with the largest, but very unlikely that the atom with 
the smallest coefficient is picked first. The idea of the proof is that OMP will only pick 'sensible' sub-supports, so we only need to ensure
concentration for a much smaller number of them. The amount of concentration needed can then be further reduced by pooling 'sensible' supports of the same type and combining probabilistic and deterministic bounds. \\
%%%% end of idea
%
%%%% proof 
\begin{proof} We use the following short hands $Q(\dico_J) = \I_d - P(\dico_J)$ as well as $r_J = Q(\dico_J)  y $
for the residual based on an index set $J$ and $x$ for the signed coefficients, $x_k: = c_k\sigma_k$. In order to better understand the various bounds of terms involving $\dico_J$, we recommend a quick familiarisation with Lemma~6.2, \cite{grrascva08}. Further we will assume w.l.o.g., that is, by reordering the dictionary matrix, that $I=\{1,\ldots, S\}=:\Sset$.
We now define the following disjoint sets for $i\leq S$ and a parameter $T > 0$, which we set to the optimal value later,
\begin{align}
 A_i &= \{1\ldots i-1\} \quad \mbox{with} \quad A_1 = \emptyset,\\
 M_i & = \{i+1, \ldots , i+T-1\}\cap \Sset,\\
 Z_i &= \{i+T, \dots ,K\} \cup \Sset^c.
\end{align}
We call a sub-support $J\subseteq \Sset$ admissible if there exist 
$i\leq S$, corresponding to the index of the missing atom with largest coefficient in magnitude, and $B\subseteq M_i$ such that $J = A_i \cup B$. We write $\Bc: = M_i/ B$. Note that an admissible sub-support $J = A_i \cup B$ contains the indices of the atoms with the first $i-1$ largest coefficients in magnitude, does not contain the index $i$, may contain indices in the support corresponding to atoms with coefficients large enough that they are likely to be picked before $i$ and is not allowed to contain any indices corresponding to atoms with too small coefficients or outside the support. \\
A sufficient condition for OMP to succeed is that it only picks admissible sub-supports.
Assuming $J$ is admissible, OMP picks another admissible support if (suff. cond.)
\begin{align} \label{suff_cond_omp}
\absip{\atom_i}{r_J} > \max_{k \in Z_i} \absip{\atom_k}{r_J},
\end{align}
which ensures that the either $i$ or some $k \in \Bc$ is chosen.
Since $J$ is admissible the residual has the form
\begin{align}\label{eq_Jres}
r_J = Q(\dico_J) y = Q(\dico_J) (x_i \atom_i + \dico_\Bc x_\Bc  + \dico_{Z_i} x_{Z_i}),
\end{align}
and we have for $i$, the index of the largest missing coefficient, 
\ifthenelse{\boolean{onecol}}
{%%% one column
%\begin{align}
%\ip{\atom_i}{r_J} &= x_i  (1 - \| P(\dico_J)\atom_i \|^2_2 ) +  \ip{\atom_i}{ [\I_d - P(\dico_J)] \dico_\Bc x_\Bc} + \ip{\atom_i}{  [\I_d - P(\dico_J)]\dico_{Z_i} x_{Z_i}}.
%\end{align}
\begin{align}\label{ipi}
\ip{\atom_i}{r_J} &= x_i  \| Q(\dico_J)\atom_i \|^2_2  +  \ip{\atom_i}{Q(\dico_J) \dico_\Bc x_\Bc} + \ip{\atom_i}{Q(\dico_J)\dico_{Z_i} x_{Z_i}}\notag\\
&= x_i  \| Q(\dico_J)\atom_i \|^2_2  +  \ip{\atom_i}{Q(\dico_J) \dico_\Bc x_\Bc} + \ip{\atom_i}{\dico_{Z_i} x_{Z_i}} - \ip{\atom_i}{P(\dico_J)\dico_{Z_i} x_{Z_i}}.
\end{align}
}
{%%% two column
\begin{align}\label{ipi}
\ip{\atom_i}{r_J} = x_i  \| Q(\dico_J)\atom_i \|^2_2  & +  \ip{\atom_i}{Q(\dico_J) \dico_\Bc x_\Bc}\notag \\
&\qquad+ \ip{\atom_i}{  Q(\dico_J)\dico_{Z_i} x_{Z_i}}.
\end{align}
}
%
%Using the identities $P(\dico_J)Q(\dico_{A_i})= P(\dico_B)Q(\dico_{A_i})$ and $Q(\dico_{J})P(\dico_{A_i})=0$, to further split the last term above into,
%\ifthenelse{\boolean{onecol}}
%{%%% one column
%\begin{align*}
% \ip{\atom_i}{  Q(\dico_J)\dico_{Z_i} x_{Z_i} }
%&= \ip{\atom_i}{ Q(\dico_J)Q(\dico_{A_i})\dico_{Z_i} x_{Z_i}} + \ip{\atom_i}{ Q(\dico_J)P(\dico_{A_i})\dico_{Z_i} x_{Z_i}} \notag \\
%&= \ip{\atom_i}{ Q(\dico_{A_i})\dico_{Z_i} x_{Z_i}} - \ip{\atom_i}{P(\dico_B)Q(\dico_{A_i})\dico_{Z_i} x_{Z_i}},
%\end{align*}
%}
%{%%% two column
%\begin{align*}
%& \ip{\atom_i}{  Q(\dico_J)\dico_{Z_i} x_{Z_i} } = \ip{\atom_i}{ Q(\dico_J)Q(\dico_{A_i})\dico_{Z_i} x_{Z_i}} \\
%&\qquad =\ip{\atom_i}{ Q(\dico_{A_i})\dico_{Z_i} x_{Z_i}} - \ip{\atom_i}{P(\dico_B)Q(\dico_{A_i})\dico_{Z_i} x_{Z_i}},
%\end{align*}
%}
%leads to 
%\ifthenelse{\boolean{onecol}}
%{%%% one column
%\begin{align} \label{ipi}
%\ip{\atom_i}{r_J} &= x_i  \| Q(\dico_J)\atom_i \|^2_2 +  \ip{\atom_i}{ Q(\dico_J)\dico_\Bc x_\Bc}\notag\\
%&\hspace{2cm}- \ip{\atom_i}{P(\dico_B)Q(\dico_{A_i})\dico_{Z_i} x_{Z_i}} + \ip{\atom_i}{ Q(\dico_{A_i})\dico_{Z_i} x_{Z_i}} . 
%\end{align}
%}
%{%%% two column
%\begin{align} \label{ipi}
%\ip{\atom_i}{r_J} &= x_i  \| Q(\dico_J)\atom_i \|^2_2 +  \ip{\atom_i}{ Q(\dico_J)\dico_\Bc x_\Bc}\notag\\
%&\hspace{1cm}- \ip{\atom_i}{P(\dico_B)Q(\dico_{A_i})\dico_{Z_i} x_{Z_i}}\notag \\
%&\hspace{2cm} + \ip{\atom_i}{ Q(\dico_{A_i})\dico_{Z_i} x_{Z_i}} . 
%\end{align}
%}
%
For $k \in Z_i$ we define $Z_i^k := Z_i/\{k\}$ and can rewrite \eqref{eq_Jres} as
\begin{align*}
r_J = Q(\dico_J) y = Q(\dico_J) ( x_i \atom_i + \dico_\Bc x_\Bc  + x_k \atom_k +\dico_{Z^k_i} x_{Z^k_i}),
\end{align*}
which leads to
\ifthenelse{\boolean{onecol}}
{%%% one column
\begin{align} \label{ipk}
\ip{\atom_k}{r_J} &=  x_i \ip{\atom_k}{Q(\dico_J)\atom_i} + \ip{\atom_k}{ Q(\dico_J) \dico_\Bc x_\Bc} \notag\\
&\hspace{2cm} + x_k  + \ip{\atom_k}{\dico_{Z_i^k} x_{Z_i^k}} - \ip{\atom_k}{ P(\dico_J)\dico_{Z_i} x_{Z_i}} .
\end{align}
}
{%%% two column
\begin{align} \label{ipk}
\ip{\atom_k}{r_J} &= x_k  \| Q(\dico_J)\atom_k \|^2_2 + \ip{\atom_k}{ Q(\dico_J) \dico_\Bc x_\Bc} \notag\\
& \quad - \ip{\atom_k}{P(\dico_B)Q(\dico_{A_i})\dico_{Z_i^k} x_{Z_i^k}}\notag \\
 & \qquad+ \ip{\atom_k}{ Q(\dico_{A_i})\dico_{Z_i^k} x_{Z_i^k}} + x_i \ip{\atom_k}{Q(\dico_J)\atom_i}. \notag 
\end{align}
}
We first bound the terms in (\ref{ipi}/\ref{ipk}) involving $Q(\dico_J)$. So for all $k\in Z_i \cup \{i\}$ we have 
\begin{align}
\| Q(\dico_J)\atom_k \|^2_2 &= 1-\| P(\dico_J)\atom_k \|^2_2 \notag \\
&=1- \ip{\dico_J^\star \atom_k}{(\dico_J^\star \dico_J)^{-1}\dico_J^\star\atom_k}\notag \\
& \geq 1- \|\dico_J^\star \atom_k\|_2 \cdot \| (\dico_J^\star \dico_J)^{-1}\|_{2,2} \cdot  \|\dico_J^\star \atom_k\|_2
\geq 1- \frac{|J| \mu^2}{1-\delta_J},
\end{align}
where we have used the bound $\| (\dico_J^\star \dico_J)^{-1}\|_{2,2} \leq (1-\delta_J)^{-1}$ from Lemma~6.2 in \cite{grrascva08}. Similarly for all $k\in Z_i$ we have
\begin{align}
\absip{\atom_k}{Q(\dico_J)\atom_i}  &= | \ip{\atom_k}{\atom_i} - \ip{\dico_J^\star \atom_k}{(\dico_J^\star \dico_J)^{-1}\dico_J^\star\atom_i} |\notag \\
& \leq  \mu + \|\dico_J^\star \atom_k\|_2 \cdot \| (\dico_J^\star \dico_J)^{-1}\|_{2,2} \cdot  \|\dico_J^\star \atom_i\|_2 \leq  \mu + \frac{|J| \mu^2}{1-\delta_J}. \label{ip_QJ}
\end{align}
and again for all $k\in Z_i \cup \{i\}$,
\ifthenelse{\boolean{onecol}}
{%%%% one column
\begin{align}
\absip{\atom_k}{ Q(\dico_J) \dico_\Bc x_\Bc} & \leq \absip{\atom_k}{\dico_\Bc x_\Bc} +  \absip{\atom_k}{ P(\dico_J) \dico_\Bc x_\Bc} \notag \\
&\leq \|\dico_\Bc^\star \atom_k\|_\infty \cdot  \| x_\Bc\|_1 + \| \dico_\Bc^\star P(\dico_J) \atom_k\|_\infty \cdot  \| x_\Bc\|_1 \notag \\
& \leq \|x_\Bc\|_1 \cdot \left( \mu + \max_{j \in \Bc} \absip{\atom_j}{ P(\dico_J) \atom_k}\right)\notag\\
& \leq \|x_\Bc\|_1 \left( \mu + \frac{|J| \mu^2 }{1-\delta_J}\right).
\end{align}
}
{%%% two column
\begin{align}
&\absip{\atom_k}{ Q(\dico_J) \dico_\Bc x_\Bc} \notag \\
&\qquad \leq \absip{\atom_k}{\dico_\Bc x_\Bc} +  \absip{\atom_k}{ P(\dico_J) \dico_\Bc x_\Bc} \notag \\
&\qquad \leq \|x_\Bc\|_1 \left( \mu + \| \dico_\Bc^\star P(\dico_J) \atom_k\|_\infty\right)\notag\\
&\qquad  \leq \|x_\Bc\|_1 \left( \mu + \frac{|J| \mu^2 }{1-\delta_J}\right).
\end{align}
}
The terms involving $P(\dico_J)$ can be bounded for all $k\in Z_i \cup \{i\}$ as
 \begin{align*}
 \absip{\atom_k}{ P(\dico_J)\dico_{Z_i} x_{Z_i}} &= \ip{\dico_J^\star \atom_k }{ (\dico_J^\star \dico_J)^{-1}\dico_J^\star\dico_{Z_i} x_{Z_i}}\\
 &\leq \| \dico_J^\star \atom_k\|_2 \cdot \|(\dico_J^\star \dico_J)^{-1}\|_{2,2} \cdot \|\dico_J^\star\dico_{Z_i} x_{Z_i}\|_2 \\
 &\leq \frac{\sqrt{|J|} \mu}{1-\delta_J}\cdot \|\dico_J^\star\dico_{Z_i} x_{Z_i}\|_2\\
 &\leq \frac{|J| \mu}{1-\delta_J}\cdot \max_{j\in J}  \absip{\atom_j}{\dico_{Z_i} x_{Z_i}}\leq \frac{|J| \mu}{1-\delta_J}\cdot \max_{j\notin Z_i}  \absip{\atom_j}{\dico_{Z_i} x_{Z_i}}.
 \end{align*}\\
 %
% involving $\Bc$. 
%For $k\in Z_i \cup \{i\}$ we have
%as well as
%\ifthenelse{\boolean{onecol}}
%{%%%% one column
%\begin{align}
%\absip{\atom_k}{P(\dico_B)Q(\dico_{A_i})\dico_{Z_i^k} x_{Z_i^k}} &\leq \| P(\dico_B)\atom_k\|_2 \cdot  \|Q(\dico_{A_i})\dico_{Z_i^k} x_{Z_i^k}\|_2 \notag\\
%& \leq \| \dico_B (\dico_B^\star \dico_B)^{-1}\|_{2,2} \cdot \| \dico_B^\star \atom_k\|_2 \cdot  \|\dico_{Z_i^k} x_{Z_i^k}\|_2\notag\\
%& \leq\frac{\mu \sqrt{|B|}}{\sqrt{1-\delta_B}} \sqrt{1+\delta_{I \cap Z_i}} \,\|x_{Z_i^k} \|_2,
%\end{align}
%}
%{%%%% two column
%\begin{align}
%&\absip{\atom_k}{P(\dico_B)Q(\dico_{A_i})\dico_{Z_i^k} x_{Z_i^k}} \notag \\
%&\hspace{2cm}\leq \| P(\dico_B)\atom_k\|_2 \|Q(\dico_{A_i})\dico_{Z_i^k} x_{Z_i^k}\|_2 \notag\\
%%& \leq \| P(\dico_B)\atom_k\|_2 \|\dico_{Z_i^k} x_{Z_i^k}\|_2
%&\hspace{2cm}  \leq\frac{\mu \sqrt{|B|}}{\sqrt{1-\delta_B}} \sqrt{1+\delta_{I \cap Z_i}} \,\|x_{Z_i^k} \|_2.
%\end{align}
%}
%where we have used the identity $\| \dico_B (\dico_B^\star \dico_B)^{-1}\|^2_{2,2} = \| (\dico_B^\star \dico_B)^{-1}\|_{2,2}$ in combination with the norm bound from Lemma~6.2 in \cite{grrascva08}.\\
From the bound above we can see that the only terms we still need to control are $\absip{\atom_j}{\dico_{Z_i} x_{Z_i}}$ for $j \notin Z_i$ and $\absip{\atom_k}{\dico_{Z^k_i} x_{Z^k_i}}$ for $k\in Z_i$. Note that for $j \notin Z_i$ we have $Z^j_i = Z_i$, meaning via Hoeffding's inequality we can estimate compactly for all $k$, 
\ifthenelse{\boolean{onecol}}
{%%%% one column
\begin{align}
\P \left(\absip{\atom_k}{\dico_{Z^k_i} x_{Z^k_i}} \geq \theta_k \right)&= \P\Big( \Big|{\sum}_{j \in Z_i^k} \ip{ \atom_k}{\atom_j }c_j \sigma_j \Big|>\theta_k\Big)\notag \\
&\leq 2\exp\left(\frac{-\theta_k^2}{ 2\sum_{j \in Z_i^k} \absip{\atom_k}{\atom_j }^2 c^2_j }\right)
\leq 2\exp\left(\frac{-\theta_k^2}{ 2\mu^2 \|c_{Z^k_i}\|^2_2}\right). \notag
\end{align}
}
{%%%% two column
\begin{align}
&\P \left(\absip{\atom_k}{ Q(\dico_{A_i})\dico_{Z_i^k} x_{Z_i^k}} \geq \theta_k \right)\notag\\
&\hspace{2cm}= \P\Big( \Big|{\sum}_{j \in Z_i^k} \ip{ \atom_k}{Q(\dico_{A_i})\atom_j }c_j \sigma_j \Big|>\theta_k\Big)\notag \\
&\hspace{2cm}\leq 2\exp\left(\frac{-\theta_k^2}{ 2\sum_{j \in Z_i} \absip{\atom_k}{Q(\dico_{A_i}) \atom_j }^2 c^2_j }\right). \notag
\end{align}
}
%To bound $\ip{ \atom_k}{Q(\dico_{A_i})\atom_j }$ we use that% \eqref{ip_QJ} with $J=A_i$. 
%\begin{align}
%\absip{Q(\dico_{A_i}) \atom_k}{\atom_j } \leq \mu + \frac{|A_i| \mu^2}{1-\delta_{A_i}}.
%\end{align}
Setting $\theta_k = 2 \theta \mu \|c_{Z_i}\|_2$ and using a union bound, we get
\begin{align}
\P \left(\exists k : \absip{\atom_k}{\dico_{Z^k_i} x_{Z^k_i}} \geq  2 \theta \mu \|c_{Z_i}\|_2 \right) \leq 2Ke^{-2\theta^2}.
\end{align}
Combining all our estimates we get that
except with probability $2Ke^{-2\theta^2}$ for all $k\in Z_i$
\ifthenelse{\boolean{onecol}}
{%%%% one column
\begin{align*}
\absip{\atom_i}{r_J} - \absip{\atom_k}{r_J}  \geq c_i \left(1- \mu - \frac{2 |J| \mu^2}{1-\delta_J}\right) - c_k - 2\|c_\Bc\|_1 &\left( \mu + \frac{|J| \mu^2 }{1-\delta_J}\right)\notag \\
&-4 \theta \|c_{Z_i}\|_2 \left(\mu + \frac{|J| \mu^2}{1-\delta_{J} }\right).
%&\geq c_i\left ( 1- \mu - \frac{2 |J| \mu^2}{1-\delta_J} - \frac{c_k}{c_i} - 2\frac{\|c_\Bc\|_1}{c_i}  \left( \mu + \frac{|J| \mu^2 }{1-\delta_J}\right)\\
\end{align*}
}
{%%%% two column
\begin{align}
&\absip{\atom_i}{r_J} - \absip{\atom_k}{r_J} \notag\\
&\hspace{5mm} \geq c_i \left(1- \mu - \frac{2 |J| \mu^2}{1-\delta_J}\right) - c_k - 2\|c_\Bc\|_1 \left( \mu + \frac{|J| \mu^2 }{1-\delta_J}\right)\notag \\
&\hspace{15mm}-2\|c_{Z_i} \|_2 \frac{\mu \sqrt{|B|}}{\sqrt{1-\delta_B}} \sqrt{1+\delta_{I}} \notag\\
&\hspace{20mm}-2 (\theta_k + \theta_i) \|c_{Z_i}\|_2 \left(\mu + \frac{|A_i| \mu^2}{1-\delta_{A_i} }\right).
%\geq c_i\left ( 1- \mu - \frac{2 |J| \mu^2}{1-\delta_J} - \frac{c_k}{c_i} - 2\frac{\|c_\Bc\|_1}{c_i}  \left( \mu + \frac{|J| \mu^2 }{1-\delta_J}\right)\\
\end{align}
}
We now determine the sets $M_i$ and $Z_i$ by choosing $T: = t \cdot\ceil{\frac{S}{\lambda}}$ to get the following bounds. For all $k\in Z_i$ we have 
\begin{align}
c_k \leq c_i \left(1- \frac{\lambda}{S}\right)^{T/t} \leq c_i \left(1- \frac{\lambda}{S}\right)^{S/\lambda} \leq c_i e^{-1}. 
\end{align}
We also have $ \|c_\Bc\|_1 \leq \|c_{M_i}\|_1\leq c_i t\ceil{\frac{S}{\lambda}}$. To bound $\|c_{Z_i}\|_2$ we note that at worst $c_{Z_i}$ consists of $t$ interleaved geometric sequences with initial value $c_i e^{-1}$ and decay factor $\alpha = 1-\lambda/S$, so we get
\begin{align*}
 \|c_{Z_i}\|_2 \leq  c_i e^{-1} \left(\frac{t}{1 - (1-\lambda/S)^2}\right)^{1/2}\leq c_i e^{-1}  \left(\frac{tS}{\lambda}\right)^{1/2}.
\end{align*}
Using $\delta_J\leq\delta_I \leq \frac12$ and setting $\theta:= \sqrt{m \log K}$, we get that except with probability $2K^{1-2m}$
for all $k \in Z_i$,
\ifthenelse{\boolean{onecol}}
{%%%% one column
\begin{align}
c_i^{-1}\left(\absip{\atom_i}{r_J} - \absip{\atom_k}{r_J} \right) &\geq 1- \mu - 4|J|\mu^2 - e^{-1} - 2t\ceil{\frac{S}{\lambda}} \left( \mu + 2|J| \mu^2\right) \notag\\
& \hspace{3.5cm}  - 4e^{-1} \sqrt{\frac{mtS\log K}{\lambda}}\left( \mu + 2 |J| \mu^2\right) \label{noiseless_cond_2}\\%\label{noisefree_cond_mid} \\
& > 0.63 - 2\left( t\ceil{\frac{S}{\lambda}} + \sqrt{\frac{mtS\log K}{\lambda}} + 1 \right)\left( \mu + 2S \mu^2\right).\notag
%&-2\|c_{Z_i} \|_2 \frac{\mu \sqrt{|B|}}{\sqrt{1-\delta_B}} \sqrt{1+\delta_{Z_i}}-
%2 (\theta_k + \theta_i) \|c_{Z_i}\|_2 \left(\mu + \frac{|A_i| \mu^2}{1-\delta_{A_i} }\right)
%\geq c_i\left ( 1- \mu - \frac{2 |J| \mu^2}{1-\delta_J} - \frac{c_k}{c_i} - 2\frac{\|c_\Bc\|_1}{c_i}  \left( \mu + \frac{|J| \mu^2 }{1-\delta_J}\right)\\
\end{align}
}
{%%%% two column
\begin{align}
&c_i^{-1}\left(\absip{\atom_i}{r_J} - \absip{\atom_k}{r_J} \right)\notag \\
&\geq 1- \mu - 4|J|\mu^2 - e^{-1} - 2t\ceil{\frac{S}{\lambda}} \left( \mu + 2|J| \mu^2\right) \notag\\
& \quad - \sqrt{6}e^{-1}  \mu t\ceil{\frac{S}{\lambda}} - 4e^{-1} \sqrt{\frac{mtS\log K}{\lambda}}\left( \mu + 2 |A_i| \mu^2\right) \notag\\%\label{noisefree_cond_mid} \\
& > 0.63 - 2\left( t\ceil{\frac{S}{\lambda}} + \sqrt{\frac{mtS\log K}{\lambda}} + 1 \right)\left( \mu + 2S \mu^2\right).\label{noiseless_cond_2}
%&-2\|c_{Z_i} \|_2 \frac{\mu \sqrt{|B|}}{\sqrt{1-\delta_B}} \sqrt{1+\delta_{Z_i}}-
%2 (\theta_k + \theta_i) \|c_{Z_i}\|_2 \left(\mu + \frac{|A_i| \mu^2}{1-\delta_{A_i} }\right)
%\geq c_i\left ( 1- \mu - \frac{2 |J| \mu^2}{1-\delta_J} - \frac{c_k}{c_i} - 2\frac{\|c_\Bc\|_1}{c_i}  \left( \mu + \frac{|J| \mu^2 }{1-\delta_J}\right)\\
\end{align}
}
In case $2 S\mu \leq 1$ the deterministic analysis holds and the theorem is trivially true. If conversely $2 S\mu \geq 1$, then $\mu + 2S\mu^2 \leq 4S\mu^2$ and so \eqref{noiseless_cond_final} implies that the last expression above is larger than zero, which further implies that
OMP will pick another admissible sub-support. Taking a union bound over all possible sets $A_i$ we get that OMP will succeed with probability at least $1 -  2SK^{1-2m}$ as long as \eqref{noiseless_cond_final} holds.
\end{proof}
%%% end proof %%%%%%
%
First note that the theorem only improves over the worst case analysis when $\lambda > t$. However, if conversely $\lambda \leq t$, then the ratio between largest and smallest coefficient is of the order $\approx e^{-1}$, so thresholding should still have a good success probability.\\
To get a better feeling for the quality of the theorem, we next specialise it to the case $t=1$, where the coefficients form a sub-geometric sequence with parameter $\alpha$, meaning $c_{i+1}/c_i \leq \alpha < 1$. In this case the theorem essentially says that OMP will recover the support except with probability $2SK^{1-2m}$ as long as $S \mu^2 \lesssim 1-\alpha$ and $S\mu^2\sqrt{m \log K} \lesssim \sqrt{1-\alpha}$. Comparing this to the condition for BP, $S\mu^2 m \log K  \lesssim 1$, for failure probability $2K^{1-2m}$, we see that OMP has the advantage that the admissible sparsity level has a milder dependence on the dictionary size and success probability while BP has the advantage of being independent of the coefficient decay. This means that each algorithm can outperform the other depending on the setting. Before confirming this in the numerical simulation in Section~\ref{sec_sim} we first have a look at the performance of OMP in a noisy setting.

\section{Noisy Case} \label{sec_noisy}
We next study partial support recovery, when the sparse signals are contaminated with noise and are modelled as
\begin{align}\label{noisy_model}
\tilde y =  y + \noise = {\sum}_k \sigma_k c_k \atom_{p(k)} + \noise,
\end{align}
with $y$ as in the previous section and $\noise$ a sub-Gaussian noise vector with parameter $\nsig$. This means that $\E(\noise) = 0$ and that for all unit vectors $v$ and $\theta>0$ the marginals $\ip{v}{\noise}$ satisfy $\E (e^{\theta \ip{v}{\noise}}) \leq e^{\theta^2 \nsig^2/2}$.\\
For Gaussian noise the parameter $\nsig$ corresponds to the standard deviation and so for normalised coefficient sequences $\|c\|_2=1$ the signal to noise ratio (SNR) is $\frac{1}{d\nsig^2}$. Similar bounds also hold in the general case, \cite{hskazh11}. In the noisy setting we clearly cannot recover coefficients below the noise level so with more decay there will be a trade off between allowing to recover more atoms and decreasing the coefficients faster. 
%%%%%%
\begin{theorem} Assume that the signals follow the model in \eqref{noisy_model} and that
for $i \leq S $ the coefficients satisfy $c_{i+t}/c_i \leq 1-\frac{\lambda}{S}$, $t,\lambda > 0$. 
Then OMP will recover an atom from the support in the first $\smallS$ steps, except with probability $4\smallS K^{1-2m}$, as long as,
\begin{align}
&\left( t\ceil{\frac{S}{\lambda}} + \sqrt{\frac{mtS\log K}{\lambda}} + 1 \right)\left( \mu + 2 s \mu^2 \right) \leq \frac{1}{10} \label{noisy_conda}\\
%&\left( t\ceil{\frac{S}{\lambda}} + \sqrt{\frac{mtS\log K}{\lambda}} + 1 \right) \left(s + t\ceil{\frac{S}{\lambda}}\right) \mu^2 \leq \frac{1}{40}\notag\\
&\mbox {and} \qquad c_s \geq 14 \nsig \sqrt{m\log K}. \label{noisy_condb}
\end{align}
\end{theorem}

%%%%% proof 
\begin{proof}
We use the same approach as before, assuming w.l.o.g. $I = \Sset$, but take into account the new expression for the residuals 
\begin{align} 
\tilde r_J = Q(\dico_J) \tilde y = Q(\dico_J)(y + \noise)= r_J + Q(\dico_J) \noise
\end{align}
and inner products 
$\ip{\atom_k}{\tilde r_J} = \ip{\atom_k}{r_J} + \ip{\atom_k}{Q(\dico_J) \noise}$.
If $J$ is an admissible sub-support, $J = A_i \cup B$ for $B\subseteq M_i$,
then in analogy to \eqref{suff_cond_omp} a sufficient condition for OMP to pick another admissible support in the noisy case is that
\begin{align} \label{suff_cond_omp_noisy}
\absip{\atom_i}{\tilde{r}_J} > \max_{k \in Z_i} \absip{\atom_k}{\tilde{r}_J}.
\end{align}
Since we have $\absip{\atom_i}{\tilde r_J} \geq \absip{\atom_i}{r_J} - \absip{\atom_i}{Q(\dico_J) \noise}$ as well as 
$\absip{\atom_k}{\tilde r_J} \leq \absip{\atom_k}{r_J} + \absip{\atom_k}{Q(\dico_J) \noise}$ for all $k\in Z_i$
the condition in \eqref{suff_cond_omp_noisy} is implied by having for all $k\in Z_i$
\begin{align}
\absip{\atom_i}{r_J} - \absip{\atom_k}{r_J} > \absip{\atom_i}{Q(\dico_J) \noise} + \absip{\atom_k}{Q(\dico_J) \noise}.
\end{align}
This means that we need to bound $\absip{\atom_k}{Q(\dico_J) \noise}$ for all $k\notin J$. Using the decomposition $J = A_i \cup B$ we can rewrite
\begin{align}
\ifthenelse{\boolean{onecol}}
{%%% one column
\ip{\atom_k}{Q(\dico_J) \noise}& = \ip{\atom_k}{Q(\dico_J)[P(\dico_{A_i}) + Q(\dico_{A_i})]\noise} \notag \\
&= \ip{\atom_k}{[\I_d - P(\dico_J)]Q(\dico_{A_i})\noise} \notag\\
&= \ip{\atom_k}{Q(\dico_{A_i})\noise} - \ip{\atom_k}{P(\dico_J)Q(\dico_{A_i})\noise} \notag \\% \label{noisesplit}
& = \ip{\atom_k}{Q(\dico_{A_i})\noise} - \ip{(\dico_J^\star \dico_J)^{-1}\dico_J^\star \atom_k}{\dico_J^\star Q(\dico_{A_i})\noise}.\label{noisesplit}
%& \hspace{-5mm}= \ip{\atom_k}{Q(\dico_{A_i})\noise} - \ip{\dico_B^\dagger \atom_k}{\dico_B^\star Q(\dico_{A_i})\noise}.\label{noisesplit}
}
{%%% two column
\ip{\atom_k}{Q(\dico_J) \noise}& = \ip{\atom_k}{Q(\dico_J)[P(\dico_{A_i}) + Q(\dico_{A_i})]\noise} \notag \\
& \hspace{-5mm}= \ip{\atom_k}{[\I_d - P(\dico_J)]Q(\dico_{A_i})\noise} \notag\\
& \hspace{-5mm}= \ip{\atom_k}{Q(\dico_{A_i})\noise} - \ip{\atom_k}{P(\dico_B)Q(\dico_{A_i})\noise}. \label{noisesplit}
%& = \ip{\atom_k}{Q(\dico_{A_i})\noise} - \ip{(\dico_B^\star \dico_B)^{-1}\dico_B^\star \atom_k}{\dico_B^\star Q(\dico_{A_i})\noise}.\label{noisesplit}
%& \hspace{-5mm}= \ip{\atom_k}{Q(\dico_{A_i})\noise} - \ip{\dico_B^\dagger \atom_k}{\dico_B^\star Q(\dico_{A_i})\noise}.\label{noisesplit}
}
\end{align}
Since $\ip{\atom_j}{Q(\dico_{A_i})\noise} = 0$ for all $j\in A_i$ this yields for all $k\notin J$ 
\begin{align}
\absip{\atom_k}{Q(\dico_J) \noise}&\leq \absip{\atom_k}{Q(\dico_{A_i})\noise} +\|(\dico_J^\star \dico_J)^{-1}\|_{2,2} \cdot \|\dico_J^\star \atom_k\|_2 \cdot \| \dico_B^\star Q(\dico_{A_i})\noise\|_2\notag\\
&\leq \absip{\atom_k}{Q(\dico_{A_i})\noise} +\frac{\mu \sqrt{|J|}}{1-\delta_J} \cdot \sqrt{|B|} \cdot \max_{j\in B} \absip{\atom_j}{Q(\dico_{A_i})\noise} \notag\\
&\leq \max_{k \notin A_i} \absip{Q(\dico_{A_i})\atom_k}{\noise}  \left(1 +\frac{\mu \sqrt{|J|\cdot |B| }}{1-\delta_J} \right). \label{noisebound}
\end{align}
To bound $\absip{Q(\dico_{A_i})\atom_k}{\noise}$ for all $k\notin A_i$ with high probability, we use the sub-Gaussian property of $\noise$ for the marginals $\ip{v_{ik}}{\noise}$, with $v_{ik} = Q(\dico_{A_i})\atom_k$. For $\smallS \leq S$ this leads to
\ifthenelse{\boolean{onecol}}
{%%% one column
\begin{align}
 \P\big(\exists i \leq \smallS,k: \absip{Q(\dico_{A_i}) \atom_k}{\noise} \geq \theta_\noise \big)
 %& \leq 2\smallS K\exp\left(-\frac{\theta_\noise^2}{2 \nsig^2\|v_{ik}\|^2_2}\right)\\
 & \leq 2\smallS K\exp\left(-\frac{\theta_\noise^2}{2 \nsig^2}\right),\notag
\end{align}
}
{%%% two column
\begin{align}
 \P\big(\exists i \leq \smallS,k: \absip{Q(\dico_{A_i}) \atom_k}{\noise} \geq \theta_\noise \big)
 %& \leq 2\smallS K\exp\left(-\frac{\theta_\noise^2}{2 \nsig^2\|v_{ik}\|^2_2}\right)\\
 & \leq 2\smallS Ke^{-\theta_\noise^2/(2 \nsig^2)},\notag
\end{align}
}
where we have used that $\|v_{ik}\|_2\leq 1$.
We substitute this bound into \eqref{noisebound} %and use the Cauchy-Schwarz inequality as well as $\dico_B^\dagger = (\dico_B^\star \dico_B)^{-1}\dico_B^\star$ 
to get, that except with probability $2\smallS Ke^{-\theta_\noise^2/(2 \nsig^2)}$ for all $k\notin J$
\begin{align}
\absip{\atom_k}{Q(\dico_J) \noise}&\leq \theta_\noise  \left(1 +\frac{\mu \sqrt{|J|\cdot |B| }}{1-\delta_J} \right). 
\end{align}
Using the intermediate bound in the noiseless case from \eqref{noiseless_cond_2} as well as the estimate above with $|B|\leq  t\ceil{\frac{S}{\lambda}}$ and $\theta_\noise = 2\nsig\sqrt{m\log K}$ we get,
except with probability $4\smallS K^{1-2m}$, for all $J=A_i \cup B $ with $|J|< s$ and all $k\in Z_i$ that
\ifthenelse{\boolean{onecol}}
{%%% one column
\begin{align}
c_i^{-1} \absip{\atom_i}{\tilde r_J} - \absip{\atom_k}{\tilde r_J} &>  0.63 - 2\left( t\ceil{\frac{S}{\lambda}} + \sqrt{\frac{mtS\log K}{\lambda}} + 1 \right)\cdot \left( \mu + 2\mu^2 s \right)\notag \\
&\hspace{4cm}  - \frac{4\nsig}{c_{\smallS}} \sqrt{m\log K} \left(1 + 2\mu \sqrt{s t\ceil{\frac{S}{\lambda}}}\right).\label{noisy_cond1}
\end{align}
}
{%%% two column
\begin{align}
c_i^{-1}& \absip{\atom_i}{\tilde r_J} - \absip{\atom_k}{\tilde r_J}\notag \\
&>  0.6 - 4\left( t\ceil{\frac{S}{\lambda}} + \sqrt{\frac{mtS\log K}{\lambda}} + 1 \right)\cdot \left( \mu + \mu^2 |J| \right) \notag \\
&\hspace{2cm} - \frac{2\nsig}{c_{\smallS}} \sqrt{m\log K} \left(1 + 2\mu t\ceil{\frac{S}{\lambda}}\right).\label{noisy_cond1}
\end{align}
}
Finally, observe that the condition in \eqref{noisy_conda} guarantees that $4\mu^2s t\ceil{\frac{S}{\lambda}} \leq \frac{1}{5}$. Thus, in order for the right hand side above to be larger than zero, which implies recovery of a correct atom in the first $s$ steps, it suffices that 
\begin{align}
\frac{4\nsig}{c_{\smallS}} \sqrt{m\log K} \left(1 + \sqrt{\frac{1}{5}}\right) \leq 0.43
\end{align}
which is guaranteed by \eqref{noisy_condb}.
\end{proof}
%%%% end of proof

\section{Numerical Simulations}\label{sec_sim}

%%%%
\ifthenelse{\boolean{onecol}}
{%%% one column
\begin{figure}[t]
\centering
\begin{tabular}{ccc}
\includegraphics[width=0.3\columnwidth]{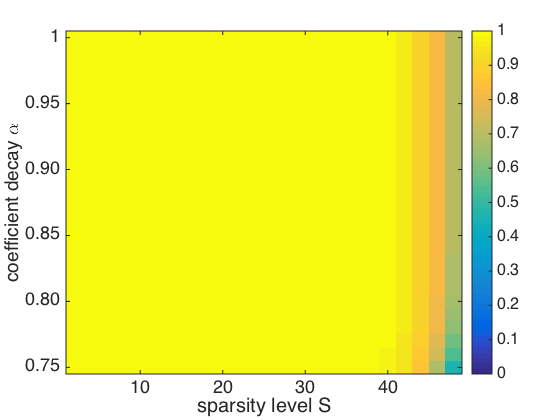}& \includegraphics[width=0.3\columnwidth]{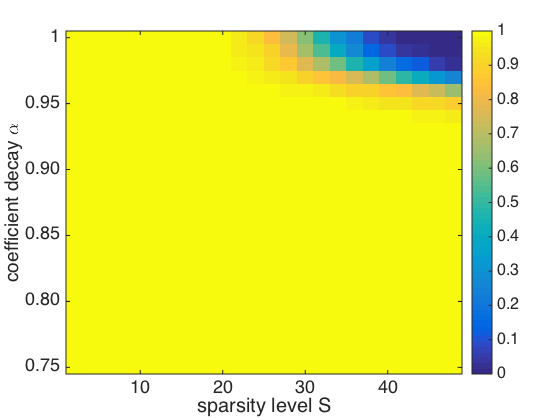} &\includegraphics[width=0.3\columnwidth]{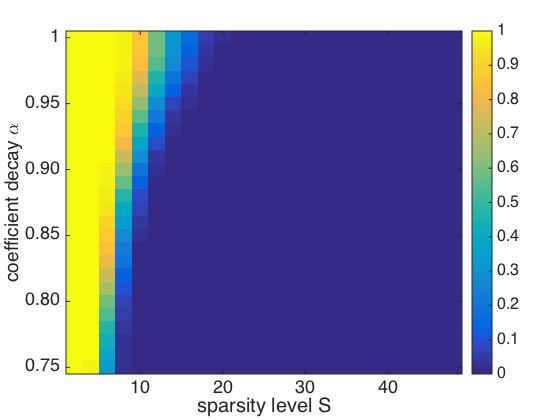}   \\
(a) & (b) & (c) \\
\includegraphics[width=0.3\columnwidth]{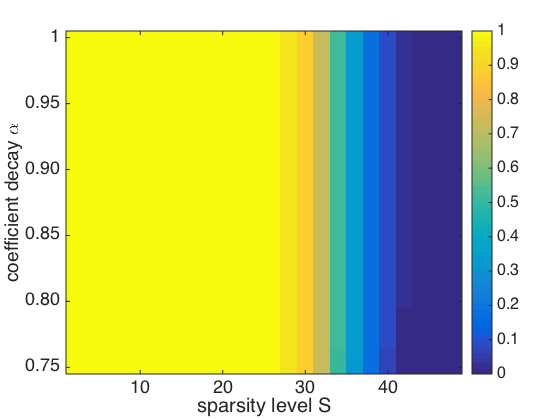}& \includegraphics[width=0.3\columnwidth]{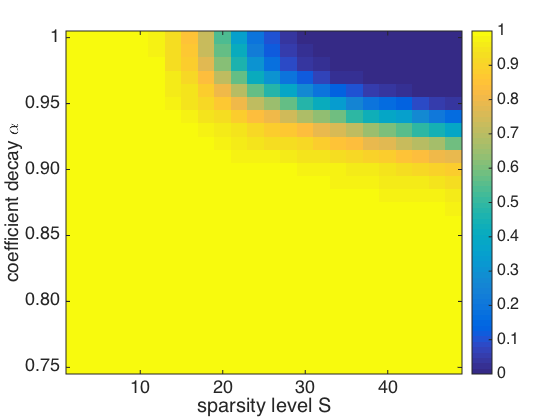}&\includegraphics[width=0.3\columnwidth]{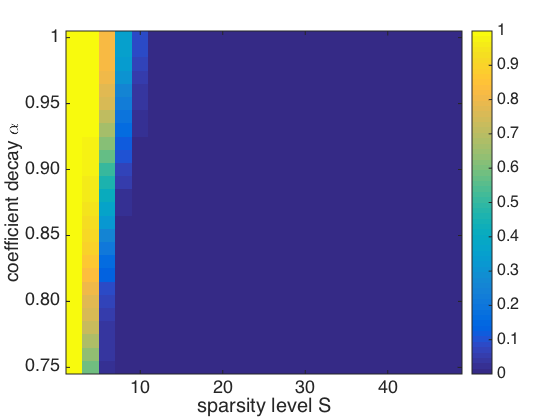}\\
(d) & (e) & (f) \\
\end{tabular}
\caption{Percentage of correctly recovered supports for noiseless signals with various sparsity and coefficient decay parameters via BP (a,d), OMP (b,e) and thresholding (c,f) in the Dirac-DCT dictionary (a,b,c) and the Dirac-DCT-random dictionary (d,e,f).\label{fig_noiseless}}
\end{figure}
}
{%%% two column
\begin{figure}
\centering
\begin{tabular}{ccc}
\includegraphics[width=0.45\columnwidth]{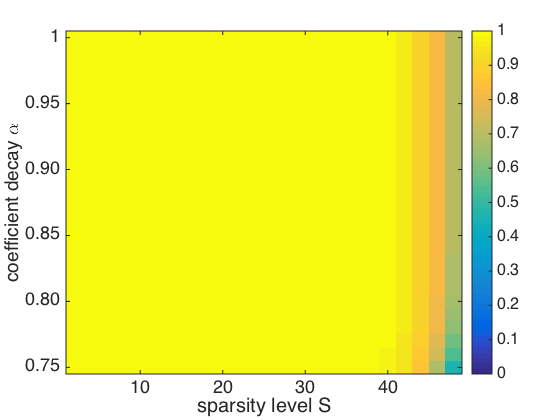}& \includegraphics[width=0.45\columnwidth]{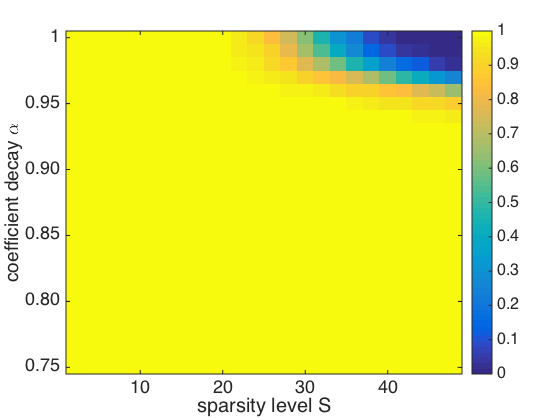}   \\
(a) & (b)\\
\includegraphics[width=0.45\columnwidth]{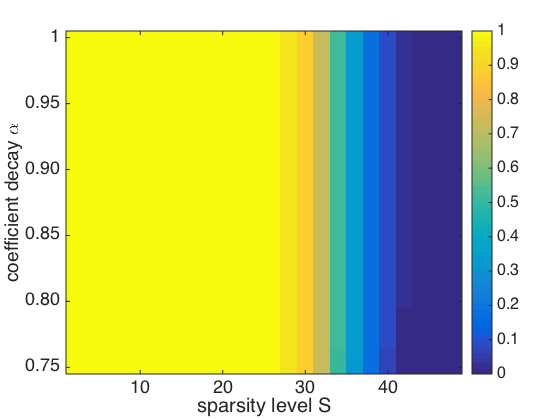}& \includegraphics[width=0.45\columnwidth]{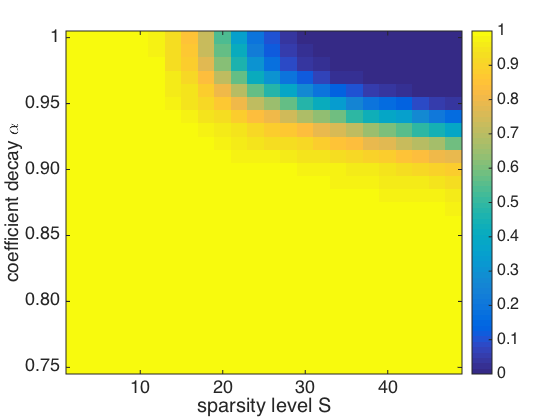}\\
(c) & (d)\\
\end{tabular}
\caption{Percentage of correctly recovered supports for noiseless signals with various sparsity and coefficient decay parameters via BP (a,c) and OMP (b,d) in the Dirac-DCT dictionary (a,b) and the Dirac-DCT-random dictionary (c,d).\label{fig_noiseless}}
\end{figure}
}

\ifthenelse{\boolean{onecol}}
{%%% one column
\begin{figure}[t]
\centering
\begin{tabular}{ccc}
\includegraphics[width=0.3\columnwidth]{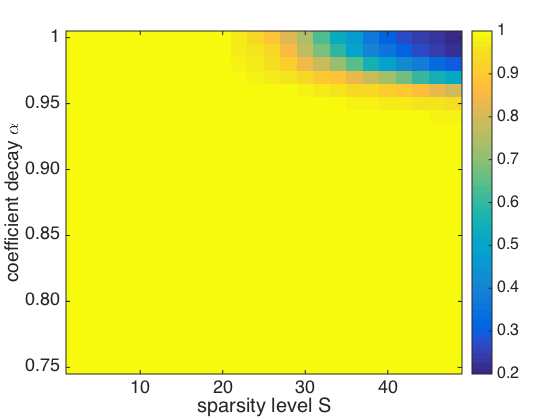} &\includegraphics[width=0.3\columnwidth]{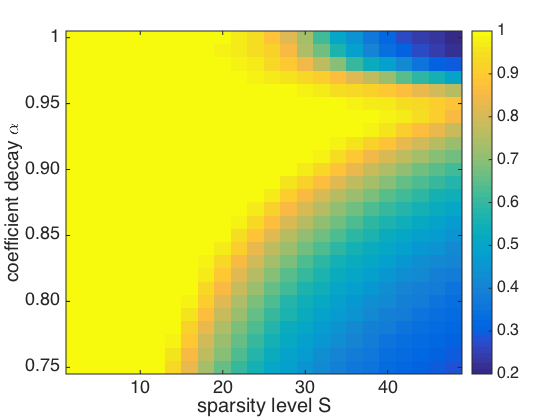}  & \includegraphics[width=0.3\columnwidth]{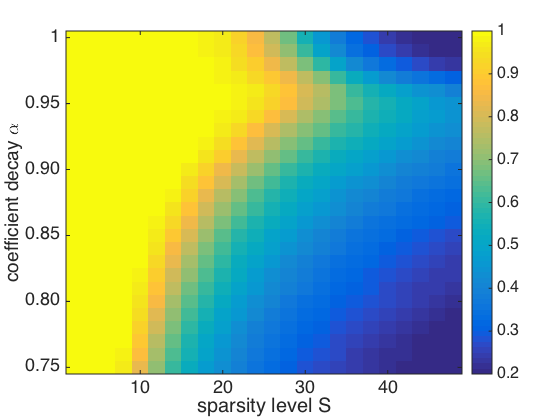}\\
(a) & (b) & (c)\\
\includegraphics[width=0.3\columnwidth]{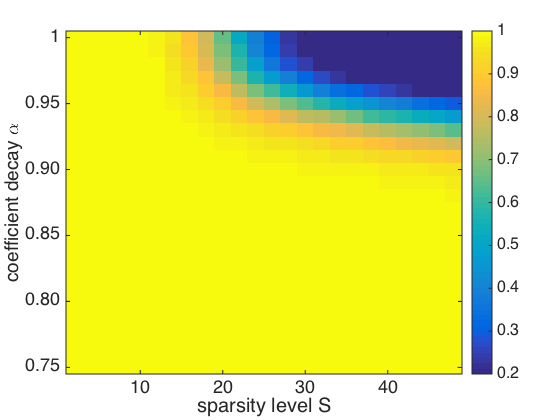} &\includegraphics[width=0.3\columnwidth]{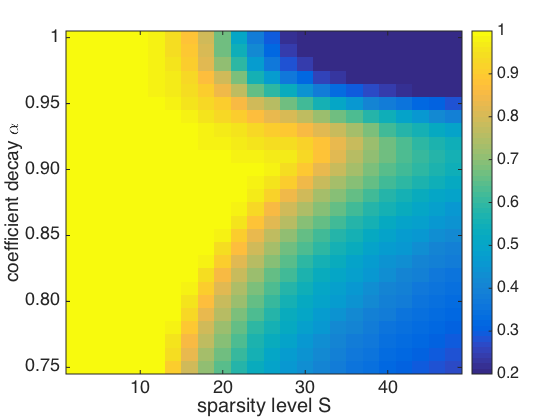}&\includegraphics[width=0.3\columnwidth]{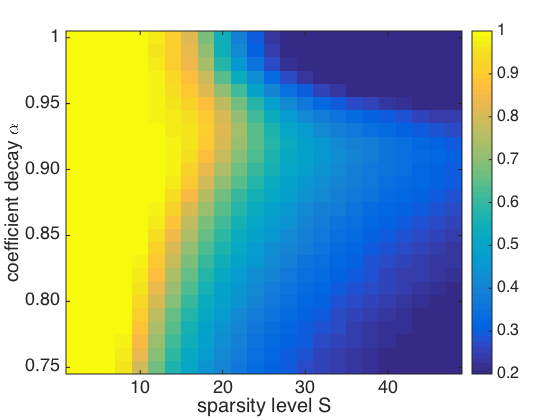}\\
(d) & (e) & (f)\\
&\includegraphics[width=0.3\columnwidth]{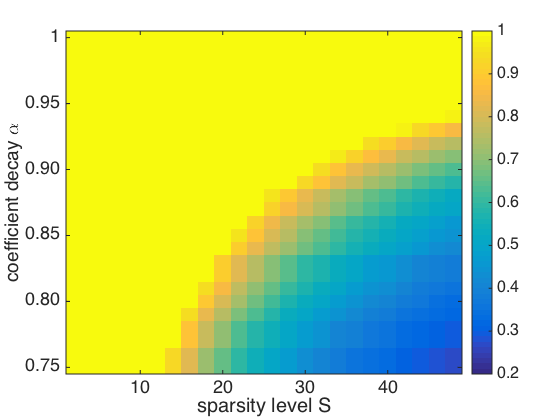}  &\includegraphics[width=0.3\columnwidth]{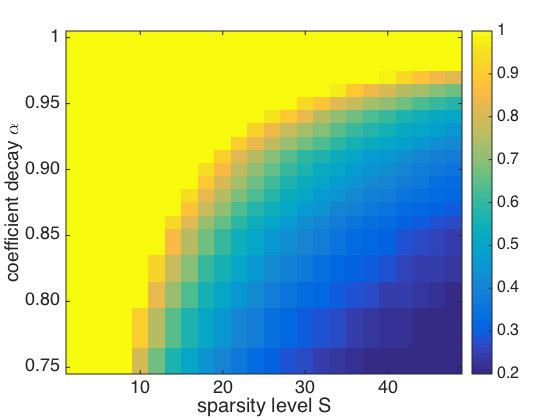}\\
&(g) & (h)
\end{tabular}
\caption{Percentage of correctly recovered atoms before recovery of first wrong atom via OMP for signals with various sparsity levels and coefficient decay parameter contaminated with no noise (a,d) or Gaussian noise corresponding to $\SNR=256$ (b,e) and $\SNR=16$ (c,f) in the
Dirac-DCT dictionary (a,b,c) and the Dirac-DCT-random dictionary (d,e,f), as well as the percentage of correctly recoverable atoms for $\SNR=256$ and $\SNR =16$ (g,h).\label{fig_noisy}}
\end{figure}
}
{%%%% two column
\begin{figure}[t!]
\centering
\begin{tabular}{cc}
\includegraphics[width=0.45\columnwidth]{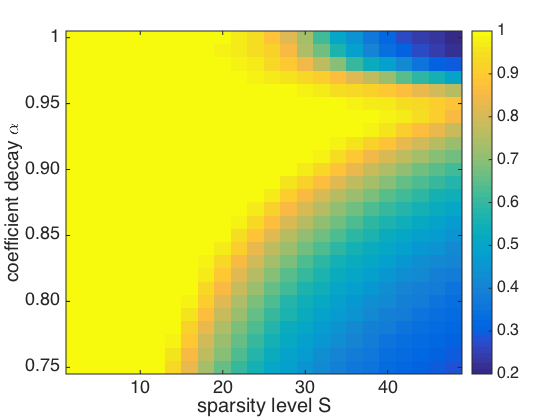}  & \includegraphics[width=0.45\columnwidth]{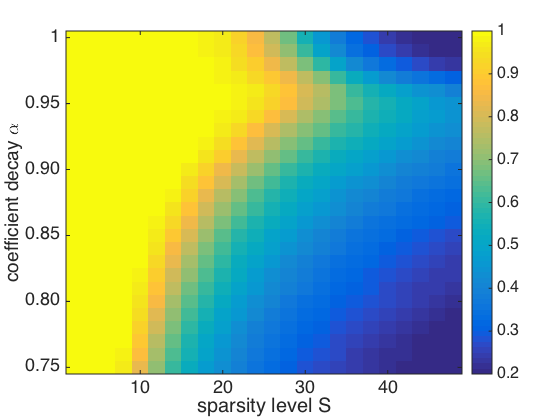}\\
(a) & (b)\\
\includegraphics[width=0.45\columnwidth]{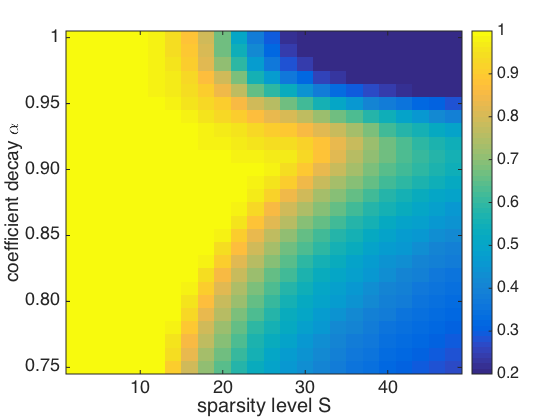}&\includegraphics[width=0.45\columnwidth]{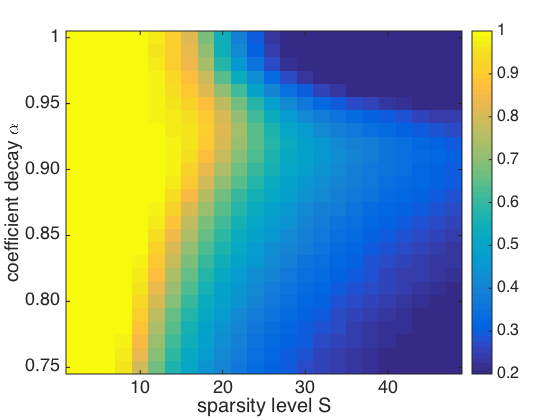}\\
(c) & (d)\\
\includegraphics[width=0.45\columnwidth]{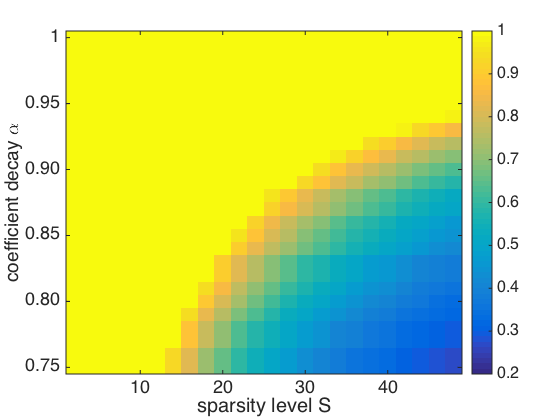}  &\includegraphics[width=0.45\columnwidth]{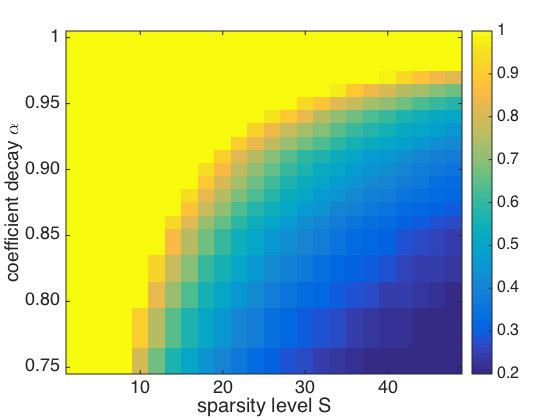}\\
(e) & (f)
\end{tabular}
\caption{Percentage of correctly recovered atoms before recovery of first wrong atom via OMP for signals with various sparsity levels and coefficient decay parameter contaminated with Gaussian noise corresponding to $\SNR=256$ (a,c) and $\SNR=16$ (b,d) in the
Dirac-DCT dictionary (a,b) and the Dirac-DCT-random dictionary (c,d), as well as the percentage of correctly recoverable atoms for $\SNR=256$ and $\SNR =16$ (e,f).\label{fig_noisy}}
\end{figure}
}

To see how well our results predict the performance of OMP, we conduct recovery experiments both with noisy and noiseless signals in $\R^d$ for $d=128$.
The signals follow the model in \eqref{noiseless_model} resp. \eqref{noisy_model}. The permutation $p$ is chosen uniformly at random and the sparse coefficients form a geometric series with parameter $\alpha$, meaning $c_i = \beta_S \alpha^i$ for $i\leq S$ and zero else, with $\beta_S$ a constant ensuring $\|c\|_2 =1$. We vary $\alpha$ between 0.75 and 1 and the sparsity level $S$ between $2$ and $48$. In case of noise, we choose $\noise_i$ i.i.d Gaussian with variance $\nsig^2 =\frac{1}{256d}$ and $\nsig^2 = \frac{1}{16d}$, corresponding to signal to noise ratios (SNR) of $256$ and $16$. As dictionaries we use the union of the Dirac and DCT bases (coherence $\mu = 0.125$) and the Dirac-DCT dictionary with additional $2d$ vectors choosing uniformly at random from the sphere ($\mu = 0.366$). \\
In the first experiment we draw $N=1000$ permutations $p$ and sign sequences $\sigma$ and for each pair $(S,\alpha)$ count how often BP, OMP and thresholding can recover the full support from the corresponding signals. From the results in Figure~\ref{fig_noiseless} we can see that the success region of OMP is indeed a union of two areas, one derived from the worst case analysis, $S\mu \lesssim 1$, and one from the average case analysis $S\mu^2 \log K \lesssim 1-\alpha $. In particular, we can see the linear dependence of the breakdown sparsity level on the parameter~$\alpha$. We can also see that the success of BP is not influenced by the coefficient decay and that, as indicated by theory, neither BP nor OMP is better in general but that each of them is better in a certain region. Finally, observe that the price you have to pay for the computational lightness of thresholding is the very limited range of parameters, where it is performing well.\\
In the second experiment we additionally draw $N$ noise-vectors to create the signals. For each signal we count how many atoms OMP identifies correctly before recovering the first incorrect atom. Figure~\ref{fig_noisy} shows the average over all $N$ realisations divided by the correct sparsity level for both dictionaries and three noise levels, as well as the relative number of recoverable atoms for the two non-zero noise levels (g, h), meaning the number of coefficients above the noise level $c_i^2 \geq 2 \nsig^2 \log K$. Comparing to the success rates in the noiseless case, we can clearly see the overlay of the two effects; for small coefficient decay we recover as many atoms as in the noiseless case, while for large decay we recover all atoms with coefficients above the noise level. 

\section{Discussion}

We have shown that OMP is successful with high probability if the coefficients exhibit decay and in such settings can even outperform BP. In particular, for geometric sequences with parameter $\alpha<1$ the admissible sparsity level scales as $S \mu^2 \log K \lesssim (1-\alpha)$. Our next goal is to extend the results to OMP using a perturbed dictionary, which is a necessary step to help tackle dictionary learning algorithms like K-SVD theoretically. We are also interested in deriving average case results for other algorithms such as stagewise OMP, \cite{dotsdrst12}, which picks more than one atom in each round, or Hard Thresholding Pursuit, \cite{fo11}, an iterative thresholding scheme. Both these algorithms can be computationally more efficient due to using less iterations and a theoretical analysis might allow the design of hybrids that automatically adapt to the decay, retain computational efficiency and allow for multiple stopping criteria.%and reducing the number of iterations where inner products need to be calculated which are expensive.
%\section*{Acknowledgements}
%This work was supported by the Austrian Science Fund (FWF) under Grant no.~Y760.
%In addition the computational results presented have been achieved (in part) using the HPC infrastructure LEO of the University of Innsbruck.
%

\acks{This work was supported by the Austrian Science Fund (FWF) under Grant no.~Y760.
The computational results presented have been achieved (in part) using the HPC infrastructure LEO of the University of Innsbruck.
Many thanks go to Simon Ruetz for our discussions of projections and his resulting identification of the bug in the original proofs\footnote{For $J = A_i \cup B$ in general $P(\dico_J)Q(\dico_{A_i})\neq P(\dico_B)Q(\dico_{A_i})$.}, which lead to this correction. Finally, thanks to both Simon Ruetz and Marie-Christine Pali for their dedicated hunt after typos, missing commas and unclear sentences.}

\bibliography{/Users/karin/Desktop/latexnotes/karinbibtex.bib}
\bibliographystyle{plain}
\end{document}